\documentclass[11pt]{article}

\usepackage[colorlinks]{hyperref}
\usepackage{amsmath} 
\usepackage{amsthm} 
\usepackage{amssymb}	
\usepackage{graphicx} 
\usepackage{multicol} 
\usepackage{multirow}
\usepackage{color}
\usepackage[dvips,letterpaper,margin=1in,bottom=1in]{geometry}
\usepackage[capitalize,noabbrev]{cleveref}
\usepackage{orcidlink} 

\usepackage[utf8]{inputenc}
\usepackage[english]{babel}

\usepackage{mathtools}
\usepackage{bbm}

\newtheorem{theorem}{Theorem}[section]

\newtheorem{lemma}[theorem]{Lemma}
\newtheorem{corollary}[theorem]{Corollary}
\newtheorem{proposition}[theorem]{Proposition}
\newtheorem{fact}[theorem]{Fact}

\newtheorem{definition}[theorem]{Definition}

\newtheorem{remark}{Remark}[section]

\newcommand{\ketbra}[2]{\ensuremath{\ket{#1}\!\bra{#2}}}

\DeclarePairedDelimiter\rbra{\lparen}{\rparen}
\DeclarePairedDelimiter\sbra{\lbrack}{\rbrack}
\DeclarePairedDelimiter\cbra{\{}{\}}
\DeclarePairedDelimiter\abs{\lvert}{\rvert}
\DeclarePairedDelimiter\Abs{\lVert}{\rVert}
\DeclarePairedDelimiter\ceil{\lceil}{\rceil}
\DeclarePairedDelimiter\floor{\lfloor}{\rfloor}
\DeclarePairedDelimiter\ket{\lvert}{\rangle}
\DeclarePairedDelimiter\bra{\langle}{\rvert}

\newcommand{\tr} {\operatorname{tr}}

\newcommand{\Var}{\mathop{\bf Var\/}}
\newcommand{\Cov}{\mathop{\bf Cov\/}}
\newcommand{\E}{\mathop{\bf E\/}}

\usepackage{algorithm}
\usepackage{algpseudocode}

\usepackage{tabularx}
\usepackage{booktabs}
\usepackage{threeparttable}

\usepackage{adjustbox}

\usepackage{subcaption}

\usepackage{tikz}
\usetikzlibrary{quantikz2}

\allowdisplaybreaks

\begin{document}

\title{Simultaneous Estimation of Nonlinear Functionals of a Quantum State}

\author{Kean Chen\thanks{Kean Chen is with the Department of Computer and Information Science, University of Pennsylvania, Philadelphia, PA 19104, USA (e-mail: \url{keanchen.gan@gmail.com}).} \and
    Qisheng Wang\thanks{Qisheng Wang is with the School of Computer Science, Shanghai Jiao Tong University, Shanghai 200240, China (e-mail: \url{QishengWang1994@gmail.com}).}\and
    Zhan Yu\thanks{Zhan Yu is with the Centre for Quantum Technologies, National University of Singapore, Singapore 119077 (e-mail: \url{yuzh.james@gmail.com}).}\and
    Zhicheng Zhang\thanks{Zhicheng Zhang is with the School of Information and Communication Technology, Griffith University, Brisbane, QLD 4111, Australia (e-mail: \url{iszczhang@gmail.com}).}
}
\date{}

\maketitle

\begin{abstract}
    We consider a fundamental task in quantum information theory, estimating the values of $\tr\rbra{\mathcal{O}\rho}$, $\tr\rbra{\mathcal{O}\rho^2}$, \dots, $\tr\rbra{\mathcal{O}\rho^k}$ for an observable $\mathcal{O}$ and a quantum state $\rho$. 
    We show that $\widetilde\Theta\rbra{k}$ samples of $\rho$ are sufficient and necessary to simultaneously estimate all the $k$ values.
    This means that estimating all the $k$ values is almost as easy as estimating only one of them, $\tr\rbra{\mathcal{O}\rho^k}$.
    As an application, our approach advances the sample complexity of entanglement spectroscopy and the virtual cooling for quantum many-body systems.
    Moreover, we extend our approach to estimating general functionals by polynomial approximation.
\end{abstract}

\newpage
\section{Introduction}
    
    Estimating the expectation value, $\tr\rbra{\mathcal{O}\rho}$, of an observable $\mathcal{O}$ with respect to a quantum state $\rho$ is a fundamental task in quantum physics, with a wide range of applications in, e.g., computation with one clean qubit \cite{KL98,She06,CM18}, approximating the Jones polynomial \cite{AJL09,SJ08}, fidelity estimation \cite{FL11,WZC+23,GP22}, quantum state tomography \cite{HHJ+17,OW16,OW17}, shadow tomography \cite{Aar18,HKP20,BO21}, and quantum machine learning \cite{HKP21}. 
To estimate expectation values with respect to thermal many-body states at lower temperatures, high-order functionals $\tr\rbra{\mathcal{O}\rho^k}$ play a key role in the virtual cooling protocol \cite{CCL+19,DTE+25}.
These types of high-order functionals were also involved in the virtual distillation protocol \cite{HMO+21,LLWF23} and localized virtual purification protocol \cite{HEY24}.

In this paper, we consider the task of estimating the values of
\begin{equation} \label{eq:k-values}
\tr\rbra{\mathcal{O}\rho},~ \tr\rbra{\mathcal{O}\rho^2},~\dots,~ \tr\rbra{\mathcal{O}\rho^k}
\end{equation}
for an observable $\mathcal{O}$ and a quantum state $\rho$. 
A direct solution to this task is to estimate each $\tr\rbra{\mathcal{O}\rho^j}$ one by one.
Specifically, each $\tr\rbra{\mathcal{O}\rho^j}$ can be estimated to within constant additive error using $O\rbra{j}$ samples of $\rho$ by the generalized SWAP test (by cyclic shift) \cite{EAO+02,Bru04} (see also \cite{CCL+19,HMO+21,SLLJ24}).
By individually estimating $\tr\rbra{\mathcal{O}\rho^j}$ each with probability $1 - 1/(3k)$ using $O\rbra{j\log\rbra{k}}$ samples of $\rho$, one can estimate all the $k$ values with probability at least $2/3$ using $\sum_{j=1}^{k} O\rbra{j\log\rbra{k}} = O\rbra{k^2\log\rbra{k}}$ samples of $\rho$. 

In sharp contrast, we show that all the $k$ values in \cref{eq:k-values} can be simultaneously estimated with little overhead when it is sufficient to estimate only one of them, $\tr\rbra{\mathcal{O}\rho^k}$.

\subsection{Main results}

We formally state this discovery as follows. 

\begin{theorem} [Simultaneous estimator, \cref{thm-2151640,thm:upper} combined] \label{thm:main}
    For any known observable $\mathcal{O}$, we can simultaneously estimate $\tr\rbra{\mathcal{O}\rho}$, $\tr\rbra{\mathcal{O}\rho^2}$, \dots, $\tr\rbra{\mathcal{O}\rho^k}$ to within additive error $\varepsilon$ using $O\rbra{k\log\rbra{k} \Abs{\mathcal{O}}^2/\varepsilon^2}$ samples of $\rho$, where $\Abs{\cdot}$ is the operator norm. 
    
    More specifically, given $n$ samples of $\rho$, we can obtain $n$ random variables $p_1, p_2, \dots, p_n$ such that for any $1 \leq k \leq n$, 
    \[
    \E\sbra{p_k} = \tr\rbra{\mathcal{O} \rho^k}, \qquad \Var\sbra{p_k} \leq \frac{2 \Abs{\mathcal{O}}^2 k}{n}.
    \]
\end{theorem}

Our estimator is actually optimal up to a logarithmic factor, as shown by the following lower bound. 

\begin{theorem} [Lower bound on estimating a single term, \cref{thm:lower-obs} restated] \label{thm:optimal}
    For any observable $\mathcal{O}$, estimating $\tr\rbra{\mathcal{O}\rho^k}$ to within additive error $\varepsilon$ requires $\Omega\rbra{k\Abs{\mathcal{O}}^2/\varepsilon^2}$ samples of $\rho$. 
\end{theorem}

Roughly speaking, \cref{thm:main} means that we can estimate all the $k$ values in \cref{eq:k-values} using only $O\rbra{k\log\rbra{k}}$ samples of $\rho$, quadratically improving the conventional $O\rbra{k^2\log\rbra{k}}$, whereas estimating the single value $\tr\rbra{\mathcal{O}\rho^k}$ already requires $\Omega\rbra{k}$ samples of $\rho$ due to \cref{thm:optimal}. 

As an implication, we can estimate $\tr\rbra{\mathcal{O}f\rbra{\rho}}$ for any polynomial $f \in \mathbb{R}\sbra{x}$. 

\begin{corollary}[Special case of \cref{corollary:poly-func}] \label{corollary:tr-O-f-rho}
    For any known observable $\mathcal{O}$ and polynomial $f \in \mathbb{R}\sbra{x}$ of degree $k$, we can estimate $\tr\rbra{\mathcal{O}f\rbra{\rho}}$ to within additive error $\varepsilon$ using $O\rbra{k\Abs{\mathcal{O}}^2\Abs{f}_1^2/\varepsilon^2}$ samples of $\rho$, where the $\ell_1$-norm $\Abs{f}_1$ is the absolute sum of the coefficients of $f$. 
\end{corollary}

The task of estimating $\tr\rbra{\mathcal{O}f\rbra{\rho}}$ in \cref{corollary:tr-O-f-rho} is an extension of estimating $\tr\rbra{\mathcal{O}\rho^k}$ with $f\rbra{x} = x^k$. 
The sample complexity in \cref{corollary:tr-O-f-rho} is optimal up to a constant factor with the hard instance where $f\rbra{x} \propto x^k$ by \cref{thm:optimal}. 
For the special case of $\mathcal{O} = I$, \cref{corollary:tr-O-f-rho} improves the result of $O\rbra{k^2 \Abs{f}_1^2/\varepsilon^2}$ in \cite{QKW24}. 
\cref{corollary:tr-O-f-rho} can be further extended to the case with multiple polynomials, i.e., estimating $\tr\rbra{\mathcal{O}f_1\rbra{\rho}}$, $\tr\rbra{\mathcal{O}f_2\rbra{\rho}}$, \dots, $\tr\rbra{\mathcal{O}f_m\rbra{\rho}}$ given $m$ polynomials $f_1$, $f_2$, \dots, $f_m$. 
See \cref{sec:extension} for more details. 

\subsection{Related work}

The problem of estimating nonlinear functionals of a quantum state has been extensively studied due to its wide applications in quantum information tasks. We summarize the comparison between our results and previous work in \cref{tab:cmp} and discuss the relevant literature below.

\begin{table}[h]
\centering
\setlength{\tabcolsep}{0.5em}
\begin{tabular}{ccc}
\toprule
\multirow{2}[2]{*}{Quantities} & \multicolumn{2}{c}{Sample Complexity} \\
\cmidrule{2-3}
& Prior Works & This Work \\
\midrule
\multirow{2}{*}{$\tr\rbra{\rho^2}$} & $O\rbra{1/\varepsilon^2}$~\cite{BCWdW01} & \multirow{2}{*}{/} \\
{} &  $\Omega\rbra{1/\varepsilon^2}$~\cite{CWLY23,GHYZ24} & \\
\midrule
\multirow{3}{*}{$\tr\rbra{\rho^k},~k \geq 3$} & $O\rbra{k/\varepsilon^2}$ \cite{EAO+02} & \multirow{3}{*}{$\Omega(k/\varepsilon^2)$} \\
& $\Omega\rbra{1/\varepsilon}$~\cite{LW25} &  \\
& $\Omega\rbra{1/\varepsilon^2}$~\cite{CW25,CLW26} & \\
\midrule
{$\tr\rbra*{f\rbra{\rho}}^\ddagger$} & $O\rbra{k^2 \Abs{f}_1^2/\varepsilon^2}$~\cite{QKW24} & ${\Theta(k \|f\|_1^2/\varepsilon^2})$ \\
\midrule
{$\cbra{\tr\rbra{\mathcal{O}\rho^j}}_{j=1}^k$} & $\widetilde{O}\rbra{k^2\Abs{\mathcal{O}}^2 /\varepsilon^2}$~\cite{SLLJ24} & ${\widetilde{\Theta}(k\|\mathcal{O}\|^2/\varepsilon^2)}$ \\
\bottomrule
\end{tabular}
\caption{Sample complexity of estimating nonlinear functionals of $\rho$. $^\ddagger$$f$ is a degree-$k$ polynomial.}
\label{tab:cmp}
\end{table}

\paragraph{Estimation of quantum state powers} The foundational work by Ekert et al.~\cite{EAO+02} introduced the generalized SWAP test, which allows for the estimation of $\tr\rbra{\rho^k}$ with additive error $\varepsilon$ using $O\rbra{k/\varepsilon^2}$ samples.  
Prior work by Quek, Kaur, and Wilde~\cite{QKW24} strengthened the generalized SWAP test by constructing a constant quantum-depth circuit for multivariate trace estimation. For the specific case of purity ($k=2$), it was shown in the context of quantum fingerprinting~\cite{BCWdW01} that $O\rbra{1/\varepsilon^2}$ samples suffice, and a corresponding lower bound $\Omega\rbra{1/\varepsilon^2}$ for purity estimation has been established in~\cite{CWLY23,GHYZ24}. For higher-order quantum state powers with $k \geq 3$, while the upper bound was established early on, finding matching lower bounds has been an active research area recently. Liu and Wang~\cite{LW25} showed a lower bound of $\Omega\rbra{1/\varepsilon}$, which was subsequently improved to $\Omega\rbra{1/\varepsilon^2}$ by~\cite{CW25,CLW26}. In this work, we close this gap by proving the lower bound of estimating a single term $\tr\rbra{\mathcal{O}\rho^k}$ even if $\mathcal{O} = I$, matching the upper bound exactly. 

\paragraph{Estimation of multiple functionals}
The task of estimating the set of nonlinear functionals $\cbra{\tr\rbra{\mathcal{O}\rho^j}}_{j=1}^k$ is more challenging.
Recent work by Shin et al.~\cite{SLLJ24} focused on resource-efficient algorithms for this task, achieving a sample complexity of $\widetilde{O}\rbra{k^2\Abs{\mathcal{O}}^2 /\varepsilon^2}$.\footnote{Throughout this paper, $\widetilde{O}\rbra{\cdot}$ and $\widetilde{\Theta}\rbra{\cdot}$ suppress polylogarithmic factors.} In contrast, our simultaneous estimator dramatically reduces this to ${\widetilde{\Theta}(k\|\mathcal{O}\|^2/\varepsilon^2)}$, demonstrating that obtaining all $k$ values is asymptotically as easy as estimating the single highest-order term $\tr\rbra{\mathcal{O}\rho^k}$. Our improvements extend to estimating functionals of the form $\tr\rbra{f\rbra{\rho}}$ where $f$ is a degree-$k$ polynomial. The prior work~\cite{QKW24} provided an approach with sample complexity $O\rbra{k^2 \Abs{f}_1^2/\varepsilon^2}$, where $\Abs{f}_1$ is the $\ell_1$-norm of the polynomial coefficients of $f$.

The aforementioned approaches are all based on the generalized SWAP test and its variants, which utilize a cyclic shift permutation to estimate a single high-order functional. 
In contrast, our work introduces a novel construction of estimators leveraging weighted permutation and symmetrization, which enables a simultaneous estimation of multiple functionals, improving the dependence on $k$ from quadratic to linear.

\paragraph{Alternative framework} We also note alternative approaches that utilize classical shadow estimation~\cite{HKP20}. Zhou and Liu~\cite{ZL24} proposed a hybrid framework for estimating nonlinear functionals $\tr\rbra{\mathcal{O}\rho^k}$ with sample complexity depending on the shadow norm of the observable $\mathcal{O}$. Additionally, Liu et al.~\cite{LLY+24} introduced auxiliary-free circuits which provide an efficient method to estimate multiple nonlinear properties. It is worth noting that the simultaneous estimation discussed in~\cite{LLY+24} targets the set of $\cbra{\tr\rbra{\mathcal{O}_l\rho^k}}_{l=1}^L$, i.e., multiple observables with a fixed state power, which is different from the task considered in this paper. While these approaches offer different trade-offs regarding quantum resources and classical post-processing, our work establishes optimal sample complexity bounds for the simultaneous estimation of nonlinear functionals.

\subsection{Recent developments}

After the work described in this paper, Shi et al.\ \cite{SJW+25} presented a time-efficient quantum algorithm for the simultaneous estimation of the nonlinear functionals $\tr\rbra{\mathcal{O}\rho}$, $\tr\rbra{\mathcal{O}\rho^2}$, \dots, $\tr\rbra{\mathcal{O}\rho^k}$. 
In particular, when $\rho$ is an $\ell$-qubit state, their approach can estimate $\tr\rbra{\rho^2}$, \dots, $\tr\rbra{\rho^k}$ on a $\rbra{2\ell+1}$-qubit quantum circuit using $O\rbra{\ell k\log\rbra{k}/\varepsilon^2}$ two-qubit gates while retaining the same sample complexity $O\rbra{k\log\rbra{k}/\varepsilon^2}$ as stated in \cref{thm:main}.
Moreover, they also extended their approach to simultaneously estimating $\tr\rbra{\mathcal{O} f_1\rbra{\rho}}$, $\tr\rbra{\mathcal{O} f_2\rbra{\rho}}$, \dots, $\tr\rbra{\mathcal{O} f_m\rbra{\rho}}$ for multiple polynomials $f_1$, $f_2$, \dots, $f_m$.
Notably, their approach can also be extended to the simultaneous estimation of $\tr\rbra{\rho\sigma}$, $\tr\rbra{\rbra{\rho\sigma}^2}$, \dots, $\tr\rbra{\rbra{\rho\sigma}^k}$, addressing Question \ref{ques1} raised in \cref{sec:discus}.
In addition to the sample access model, the query complexity of nonlinear functional estimation was studied in \cite{cllg-15kd,Wan25}.

Our hard instance used for the sample complexity lower bound in \cref{thm:optimal} was later used to show the quantum query complexity lower bound for estimating the Tsallis entropy of integer order in \cite{Wan25}, using the quantum query complexity lower bound for distinguishing probability distributions in \cite{Bel19}. 
Later, an alternative proof (while using the same hard instance) was provided in \cite{CWZ25} by the quantum sample-to-query lifting method \cite{WZ25a,WZ25b,TWZ25}. 

\subsection{Organization of this paper}

The construction of the simultaneous estimators in \cref{thm:main} will be presented in \cref{sec:simul}, and their sample complexity will be analyzed in \cref{sec:upper}. 
The sample complexity lower bounds will be given in \cref{sec:lower}. 
Applications of our simultaneous estimators will be discussed in \cref{sec:app}. 
Finally, a brief discussion will be given in \cref{sec:discus} with several open questions. 

\section{Technical Overview}

\subsection{The simultaneous estimators}\label{sec-4271228}

Now we introduce the main idea of constructing the simultaneous estimators in \cref{thm:main}. 

Suppose $\rho$ is a quantum state in a finite-dimensional Hilbert space $\mathcal{H}$. 
Let $\mathcal{L}\rbra{\mathcal{H}}$ denote the set of linear operators on $\mathcal{H}$.
Let $\mathfrak{S}_n$ be the symmetric group on $\{1,\ldots,n\}$ and for each permutation $\pi\in\mathfrak{S}_n$, let $U_\pi$ be the unitary representation of $\pi$ on $\mathcal{H}^{\otimes n}$, i.e., $U_\pi\ket{\alpha_1}\cdots\ket{\alpha_n}=\ket{\alpha_{\pi^{-1}(1)}}\cdots\ket{\alpha_{\pi^{-1}(n)}}$. 
Then, we exploit the permutation symmetry of $\rho^{\otimes n}$ in designing the observables, which is a basic idea widely used in quantum information theory (see, e.g., \cite{alicki1988symmetry,hayashi1998asymptotic,KW01,MdW16,BOW19}).
Specifically, let $\Phi$ be the symmetrization map: $\Phi(M)= \frac{1}{n!}\sum_{\pi\in\mathfrak{S}_n} U_{\pi}M U_{\pi}^\dag$, where $M\in\mathcal{L}(\mathcal{H}^{\otimes n})$.
Then, given an observable $\mathcal{O}\in\mathcal{L}(\mathcal{H})$, for each $1\leq k\leq n$, we define: 
\begin{equation}\label{eq-3181224}
\mathcal{O}_k\coloneqq \Phi\rbra*{U_{s_k}\cdot \rbra*{\mathcal{O}\otimes I^{\otimes n-1}}},
\end{equation}
where \(s_k\in\mathfrak{S}_n\) is the cyclic shift permutation on the first \(k\) elements, i.e., $s_k(i)= \rbra{i \bmod k} + 1$ for $i\leq k$ and $s_k(i)=i$ for $k< i\leq n$. We will show that those $\mathcal{O}_k\in\mathcal{L}(\mathcal{H}^{\otimes n})$ are valid observables and can be used to estimate $\tr(\mathcal{O}\rho^k)$ efficiently.

To study the estimators $\mathcal{O}_k$, it is more convenient to treat $U_{s_k}\cdot\rbra{\mathcal{O}\otimes I^{\otimes n-1}}$ in \cref{eq-3181224} as a \textit{weighted permutation} (see \cref{def:weighted-perm}). Specifically, a weighted permutation of degree $n$ is a tuple $(\pi,w)$, where $\pi\in\mathfrak{S}_n$ is a permutation and $w=[w_1,\ldots,w_n]\in\mathbb{Z}^n_{\geq 0}$ is a length-$n$ vector with non-negative integer entries representing the weights. We may directly use $\pi w$ to denote $(\pi,w)$ without causing confusion. We define the matrix representation of $\pi w$ as
\begin{equation*}
\mu(\pi w)=U_\pi\cdot (\mathcal{O}^{w_1}\otimes \mathcal{O}^{w_2}\otimes\cdots\otimes \mathcal{O}^{w_n}).
\end{equation*}
Let $\mathfrak{W}_n$ denote the set of all weighted permutations of degree $n$. Then, $\mathfrak{W}_n\cong \mathfrak{S}_n\ltimes \mathbb{Z}_{\geq 0}^n$ forms a monoid, where the multiplication in $\mathfrak{W}_n$ coincides with the matrix multiplication on its matrix representation, i.e., $\mu(\pi w \cdot \pi' w')=\mu(\pi w)\cdot \mu(\pi' w')$. Additionally, $\mathfrak{W}_n$ is also equipped with an involution ``$\dag$'' which coincides with the Hermitian transpose on its matrix representation, i.e., $\mu((\pi w)^\dag)=\mu(\pi w)^\dag$. Some examples are shown in \cref{fig-3260329}. We can identify each $\pi\in\mathfrak{S}_n$ as $\pi \mathbf{0}\in\mathfrak{W}_n$ (where $\mathbf{0}\in\mathbb{Z}_{\geq 0}^n$ is the zero vector). Then, we interpret $\Phi$ as a symmetrization map acting on the monoid ring $\mathbb{C}\mathfrak{W}_n$ (the set of finite formal sums of elements in $\mathfrak{W}_n$ with complex coefficients), i.e., $\Phi(\mathcal{X})=\frac{1}{n!}\sum_{\pi\in\mathfrak{S}_n} \pi \mathcal{X}\pi^{-1}$ for any $\mathcal{X}\in\mathbb{C}\mathfrak{W}_n$. Thus, our estimators can be written as $\mathcal{O}_k=\mu(\Phi(s_k e_1))$, where $e_1=[1,0,\ldots,0]\in\mathbb{Z}_{\geq 0}^n$. 

Then, consider the orbit of a weighted permutation $X\in\mathfrak{W}_n$ under the conjugation action of $\mathfrak{S}_n$, i.e., $\{\pi X\pi^{-1} \,|\, \pi\in\mathfrak{S}_n\}$. It turns out that each orbit corresponds to a generalized cycle type (see \cref{def:weighted-cycle-type}), referred to as a weighted cycle type, where each transition within the cycle is associated with a weight $w_i$. \cref{fig-3260334} shows the weighted cycle types corresponding to the example in \cref{fig-3262240}. We remark that the involution does not necessarily preserve the weighted cycle type in general, which behaves differently 
from the ordinary cycle types of permutations. 
Nevertheless, 
it does preserve the weighted cycle type when the total weight $|w|_1=\sum_{i}w_i\leq 2$.

\begin{proposition}[\cref{lemma-4121756}]
    For $\pi w \in \mathfrak{W}_n$, if $\abs{w}_1 \leq 2$, then $\pi w$ and $\rbra{\pi w}^\dag$ have the same weighted cycle type. 
\end{proposition}

This further leads to the following two key properties of $\mathcal{O}_k$.

\begin{enumerate}
    \item Hermiticity: The estimator \(\mathcal{O}_k\) is Hermitian and thus is a valid observable. Indeed, we have $\mathcal{O}_k^\dag=\mu(\Phi((s_k e_1)^\dag))$. Since the total weight $|e_1|_1=1\leq 2$, $(s_k e_1)^\dag$ and $s_k e_1$ have the same weighted cycle type (thus the same orbit). Note that $\Phi(s_k e_1)$ produces the average of elements in the orbit of $s_k e_1$ and thus equals $\Phi((s_k e_1)^\dag)$. Therefore, $\mathcal{O}_k^\dag=\mathcal{O}_k$.
    \item Commutativity:
The observables $\{\mathcal{O}_i\}_{i=1}^n$ pairwise commute, thereby enabling simultaneous measurements. 
To prove commutativity, it suffices to show that the coefficient of each  weighted permutation in $\mathcal{O}_i\mathcal{O}_j$ coincides with that in $\mathcal{O}_j\mathcal{O}_i=(\mathcal{O}_i\mathcal{O}_j)^\dag$. Since each weighted permutation occurring in $\mathcal{O}_i\mathcal{O}_j$ has total weight $|w|_1=2$, the involution $\dag$ preserves its weighted cycle type. Therefore, for any weighted cycle type $\tau$, the sum of coefficients of all weighted permutations of type $\tau$ in $\mathcal{O}_i\mathcal{O}_j$ is the same as that in $(\mathcal{O}_i\mathcal{O}_j)^\dag$.
On the other hand, since all $\mathcal{O}_i$ are permutation-invariant (i.e., $U_\pi \mathcal{O}_i U_\pi^\dag=\mathcal{O}_i$ for any $\pi\in\mathfrak{S}_n$), $\mathcal{O}_i\mathcal{O}_j$ is permutation-invariant. Thus all weighted permutations of type $\tau$ must share the same coefficient in $\mathcal{O}_i\mathcal{O}_j$. The same argument also works for $(\mathcal{O}_i\mathcal{O}_j)^\dag$. Then, we conclude that the coefficient of each weighted permutation in $\mathcal{O}_i\mathcal{O}_j$ is exactly the same as that in $(\mathcal{O}_i\mathcal{O}_j)^\dag$.
More details can be found in \cref{lemma-2110242}.
\end{enumerate}

\begin{figure}[!t]
    \centering
    \begin{subfigure}[b]{0.6\linewidth}
    \includegraphics[width=1.0\linewidth]{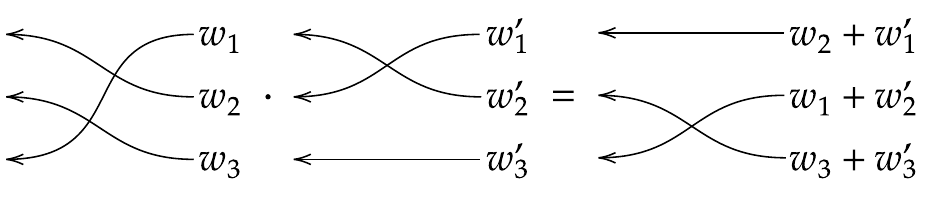}
    \caption{Multiplication.}
    \end{subfigure}
    \\
    \vspace{2mm}
    \begin{subfigure}[b]{0.6\linewidth}
    \includegraphics[width=1.0\linewidth]{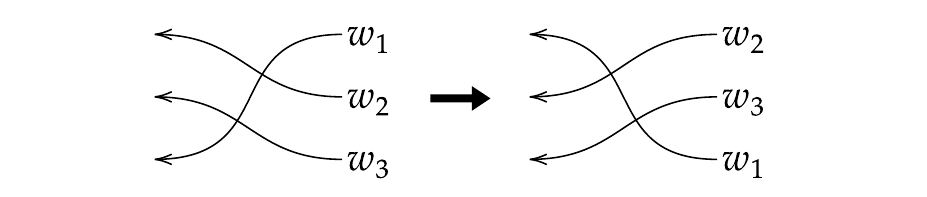}
    \caption{Involution.}\label{fig-3262240}
    \end{subfigure}
    \caption{Diagrammatic illustration of operations on the weighted permutations.}
    \label{fig-3260329}
\end{figure}

\begin{figure}[!t]
    \centering
    \includegraphics[width=0.65\linewidth]{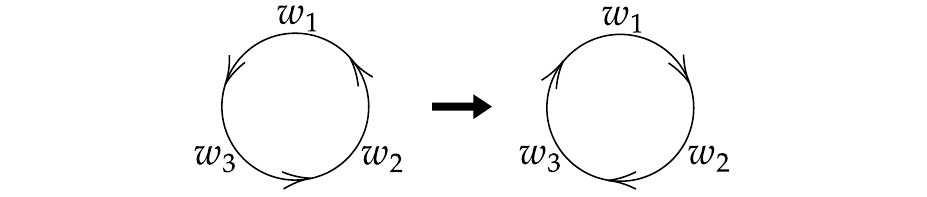}
    \caption{Weighted cycle types before/after the involution.}
    \label{fig-3260334}
\end{figure}

Given the Hermiticity and commutativity of the estimators, our next step is to study the expectation and variance of each estimator. 
First, it is easy to see that $\mathcal{O}_k$ is an unbiased estimator for $\tr(\mathcal{O}\rho^{k})$, i.e., $\E_{\rho^{\otimes n}}[\mathcal{O}_k]=\tr(\mathcal{O}_k\rho^{\otimes n})=\tr(U_{s_k}(\mathcal{O}\otimes I^{\otimes n-1})\rho^{\otimes n})=\tr(\mathcal{O}\rho^k)$. 
Then, to bound the variance, we treat our estimator $\mathcal{O}_k$ as a symmetrization of the natural estimator $T_k$ for the quantity $\tr(\mathcal{O}\rho^k)$, which is
\begin{align*}
    T_k=\frac{1}{2\floor{n/k}}\sum_{i=0}^{\floor{n/k}-1} &\Big(\mu(s_{(ik+1,\ldots,ik+k)}e_{ik+1})
    +\mu(s_{(ik+1,\ldots,ik+k)}e_{ik+1})^\dag\Big),
\end{align*}
where $s_J$ is the cyclic shift permutation on the sequence $J$ and $e_{i}=[0,\ldots,1,\ldots,0]\in\mathbb{Z}_{\geq 0}^n$ is the vector with a $1$ in the $i$-th place and $0$'s elsewhere. 
By the Kadison-Schwarz inequality~\cite{kadison1952generalized} (cf.\ \cite{BOW19}), we can show that the variance of the symmetrized estimator $\mathcal{O}_k$ is no more than that of $T_k$, i.e.,
\begin{align}
\Var[\mathcal{O}_k]\leq \Var[T_k] \leq \frac{2k\|\mathcal{O}\|^2}{n}. \nonumber
\end{align}
Then, by Chebyshev's inequality and the median trick, it suffices to use $O(k\log(k)\|\mathcal{O}\|^2/\varepsilon^2)$ samples of $\rho$ to give an estimate of $\tr(\mathcal{O}\rho^k)$ to within additive error $\varepsilon$ with probability at least $1-1/(3k)$. Similarly, these samples can be reused (due to the commutativity of $\{\mathcal{O}_i\}_{i=1}^k$) to obtain an estimate of each of $\tr(\mathcal{O}\rho^{k-1}),\ldots,\tr(\mathcal{O}\rho)$ with probability $1-1/(3k)$. 
Therefore, the success probability of the whole process is at least $(1-1/(3k))^k\geq 2/3$.

\subsection{Lower bounds} 
To prove the matching lower bound in \cref{thm:optimal}, we first reduce the problem to the case of $\Abs*{\mathcal{O}}=1$ and $\bra{0}\mathcal{O}\ket{0}=1$ by rescaling.
In this case, we can obtain a lower bound $\Omega\rbra*{k/\varepsilon^2}$ by reducing from the lower bound for quantum state discrimination.
Then, we find a new hard instance not well-known in the literature. Specifically, the pair of quantum states $\rho_+$ and $\rho_-$ for the discrimination task is constructed as:
\[
\rho_\pm = \rbra*{1-\frac{1}{k} \pm \frac{\varepsilon}{k}} \ketbra{0}{0} + \rbra*{\frac{1}{k}\mp\frac{\varepsilon}{k}} \ketbra{1}{1}, 
\]
where $\varepsilon \in \rbra{0, 1}$. 
For sufficiently small $\varepsilon > 0$, it can be verified that $\tr\rbra{\mathcal{O}\rho_+^k} - \tr\rbra{\mathcal{O}\rho_-^k} \geq \Omega\rbra{\varepsilon}$.
Therefore, any estimator for $\tr\rbra{\mathcal{O}\rho^k}$ to within additive error $\Theta\rbra{\varepsilon}$ can be used to distinguish $\rho_+$ and $\rho_-$. 
On the other hand, the infidelity between $\rho_+$ and $\rho_-$ is bounded by $\gamma = 1 - \mathrm{F}\rbra{\rho_+, \rho_-} \leq O\rbra{\varepsilon^2/k}$. 
By the Helstrom-Holevo bound \cite{Hel67,Hol73} (cf.\ \cite{Wil13,Hay17}), the sample complexity of distinguishing $\rho_+$ and $\rho_-$ is lower bounded by $\Omega\rbra{1/\gamma} = \Omega\rbra{k/\varepsilon^2}$, which easily leads to \cref{thm:optimal} by rescaling.

Note that another consequence of \cref{thm:optimal} is the optimality of the generalized SWAP test \cite{EAO+02,Bru04} (see also \cite{CCL+19,HMO+21,SLLJ24}) for estimating a single term $\tr\rbra{\mathcal{O}\rho^k}$. In comparison,
our \cref{thm:main} implies that simultaneously estimating all the $k$ values in \cref{eq:k-values} only incurs an $O\rbra*{\log(k)}$ factor.

\section{Simultaneous Estimators} \label{sec:simul}

In this section, we prove the following result.
\begin{theorem}\label{thm-2151640}
Suppose \(\mathcal{O}\in\mathcal{L}(\mathcal{H})\) is an observable. Given \(n\) samples of an unknown state \(\rho\), there is an algorithm that outputs a list $p_1, p_2, \ldots, p_n $ such that for any \(1\leq k\leq n\), we have
\[\E[p_k]=\tr(\mathcal{O} \rho^k) \quad\quad\textup{and}\quad\quad \Var[p_k]\leq \frac{2k\|\mathcal{O}\|^2}{n},\]
where \(\|\mathcal{O}\|\) is the operator norm of \(\mathcal{O}\).
\end{theorem}

The proof is based on the idea we outlined in \cref{sec-4271228} with full details provided here.
First, we formally define the weighted permutations and weighted cycle type in \cref{sec-4271205} and \cref{sec-4271212}. Then, we use these notations to construct our simultaneous estimators in \cref{sec-2111442} and bound the variance in \cref{sec-4271217}.
The proof of \cref{thm-2151640} is summarized in \cref{sec-4131102}.

\subsection{Weighted permutations}\label{sec-4271205}
First, we define the weighted permutation.
\begin{definition}[Weighted permutation] \label{def:weighted-perm}
A weighted permutation of degree $n$ is a tuple $(\pi,w)$ where $\pi\in\mathfrak{S}_n$ is a permutation and $w=[w_1,\ldots,w_n]\in\mathbb{Z}_{\geq 0}^n$ is a length-$n$ vector with non-negative integer entries representing the weights. 
We use $\mathfrak{W}_n$ to denote the set of all weighted permutations of degree $n$. 
\end{definition}
Then, we define the total weight of a weighted permutation $X\in\mathfrak{W}_n$.
\begin{definition}
For any $X=(\pi,w)\in\mathfrak{W}_n$, we define the total weight of $X$ as $|X|=|w|_1=\sum_i w_i$. 
\end{definition}
Then, we define the multiplication and involution on $\mathfrak{W}_n$.
\begin{definition}
We define two operations on $\mathfrak{W}_n$:
\begin{itemize}
    \item Multiplication \textup{``}$\cdot$\textup{''}:
    \[(\pi, w)\cdot (\pi', w') = (\pi \pi', w_{\pi'} + w'  ),\]
    where $w_{\pi'}=[w_{\pi'(1)},w_{\pi'(2)},\ldots,w_{\pi'(n)}]\in\mathbb{Z}_{\geq 0}^n$.
    \item Involution \textup{``}$\dag$\textup{''}:
    \[(\pi,w)^\dag= (\pi^{-1},w_{\pi^{-1}}),\]
    where $w_{\pi^{-1}}=[w_{\pi^{-1}(1)},w_{\pi^{-1}(2)},\ldots,w_{\pi^{-1}(n)}]\in\mathbb{Z}_{\geq 0}^n$.
\end{itemize}
\end{definition}
Some examples are shown in \cref{fig-3260329}.

\begin{remark}
For convenience, we may directly use $\pi w$ to denote the weighted permutation $(\pi,w)$. 
\end{remark}
It is easy to check the following properties.
\begin{fact}\label{fact-4121934}
For any $X,Y,Z\in\mathfrak{W}_n$,
\begin{itemize}
    \item $X\cdot (Y\cdot Z) = (X\cdot Y) \cdot Z$,
    \item $(I\mathbf{0})\cdot X=X\cdot (I\mathbf{0})= X$, where $I\in\mathfrak{S}_n$ is the identity permutation and $\mathbf{0}\in\mathbb{Z}_{\geq 0}^n$ is the zero vector,
    \item $(X\cdot Y)^\dag=Y^\dag \cdot X^\dag$.
\end{itemize}
\end{fact}
In fact, $\mathfrak{W}_n$ can be viewed as a monoid $\mathfrak{S}_n\ltimes \mathbb{Z}_{\geq 0}^n$ with an additional involution operation, where $\mathbb{Z}_{\geq 0}^n$ is an abelian monoid with natural addition. On the other hand, we can also see the following properties.
\begin{fact}\label{fact-4130216}
For any $X,Y\in\mathfrak{W}_n$,
\begin{itemize}
    \item $|X\cdot Y|= |X|+|Y|$,
    \item $|X^\dag|=|X|$.
\end{itemize}
\end{fact}

Then, we consider the matrix representation for $\mathfrak{W}_n$.
\begin{definition}\label{def-4120344}
Given a Hermitian operator $\mathcal{O}\in\mathcal{L}(\mathcal{H})$, we define the matrix representation $\mu: \mathfrak{W}_n\rightarrow \mathcal{L}(\mathcal{H}^{\otimes n})$ based on $\mathcal{O}$ as 
$$\mu(\pi w)=U_\pi\cdot (\mathcal{O}^{w_1}\otimes \mathcal{O}^{w_2}\otimes \cdots \otimes \mathcal{O}^{w_n}),$$
where $U_\pi$ denotes the permutation operator acting on the space \(\mathcal{H}^{\otimes n}\), i.e., $U_\pi \cdot \ket{\psi_1}\cdots \ket{\psi_{n}}=\ket{\psi_{\pi^{-1}(1)}}\cdots \ket{\psi_{\pi^{-1}(n)}}$.
\end{definition}
We can easily see that $\mu$ is a valid matrix representation for $\mathfrak{W}_n$ through the following properties.
\begin{fact}\label{fact-4121936}
For any $X,Y\in\mathfrak{W}_n$,
\begin{itemize}
\item $\mu(X\cdot Y)=\mu(X)\cdot \mu(Y)$,
\item $\mu(X)^\dag = \mu(X^\dag)$.
\end{itemize}
\end{fact}

Then we extend our definition to the monoid ring $\mathbb{C}\mathfrak{W}_n$, which is the set of formal sums \(\sum_{X\in\mathfrak{W}_n} b_X X\), where $b_X\in\mathbb{C}$ and $b_X=0$ for all but only finitely many $X$. The involution is extended anti-linearly on $\mathbb{C}\mathfrak{W}_n$, i.e., $(\sum_X b_X X)^\dag=\sum_X b_X^* X^\dag$, where $b_X^*$ is the complex conjugate of $b_X$.
The matrix representation $\mu$ can be extended naturally on $\mathbb{C}\mathfrak{W}_n$, i.e., 
$\mu(\sum_X b_X X)=\sum_X b_X\mu(X)$. 
Then, one can easily check that the properties in \cref{fact-4121934} and \cref{fact-4121936} also hold on $\mathbb{C}\mathfrak{W}_n$. For convenience, we will use the following notation.
\begin{definition}\label{def-4130359}
For every $X \in \mathfrak{W}_n$, define $c_X \colon \mathbb{C}\mathfrak{W}_n \to \mathbb{C}$ such that for any $\sum_{Y} b_{Y} Y\in\mathbb{C}\mathfrak{W}_n$,
\[c_X\left(\sum_{Y} b_{Y} Y\right)= b_X.\]  
\end{definition}

\subsection{Weighted cycle type}\label{sec-4271212}
There is a natural inclusion map from $\mathfrak{S}_n$ to $\mathfrak{W}_n$ by identifying each $\pi\in\mathfrak{S}_n$ as $\pi \mathbf{0}\in\mathfrak{W}_n$ (where $\mathbf{0}\in\mathbb{Z}_{\geq 0}^n$ is the zero vector).

Then, consider the conjugation action of $\mathfrak{S}_n$. The orbit of a weighted permutation $X\in\mathfrak{W}_n$ w.r.t. the conjugation action of $\mathfrak{S}_n$ is the set $\{\pi X\pi^{-1} \,|\, \pi\in\mathfrak{S}_n\}$.
Such an orbit can be represented by a generalized cycle type, referred to as a weighted cycle type, where each transition within the cycle is associated with a weight $w_i$.
\begin{definition}[Weighted cycle type] \label{def:weighted-cycle-type}
A weighted cycle type of degree $n$ is a disjoint union of directed cycle graphs with $n$ vertices, where each edge $e$ is assigned a weight $w_e\in\mathbb{Z}_{\geq 0}$. We use $\tau\vdash n$ to denote that $\tau$ is a weighted cycle type of degree $n$.

Two weighted cycle types $\tau,\nu\vdash n$ are considered the same if there is a bijection $f$ between the vertex sets of $\tau$ and $\nu$ such that there exists an edge with weight $w_{v_1v_2}\in\mathbb{Z}_{\geq 0}$ from $v_1$ to $v_2$ if and only if there exists an edge with weight $w_{v_1v_2}$ from $f(v_1)$ to $f(v_2)$.
\end{definition}
An example of weighted cycle type is shown in \cref{fig-4110128}.
\begin{figure}[ht]
    \centering
    \begin{subfigure}[b]{0.3\linewidth}
        \centering
    \includegraphics[width=1.0\linewidth]{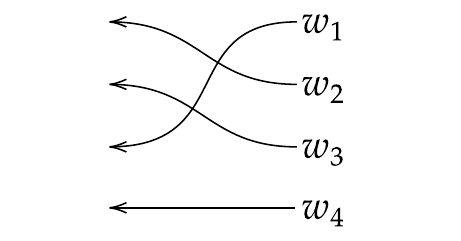}
    \caption{A weighted permutation.}\label{fig-4110125}
    \end{subfigure}
    \begin{subfigure}[b]{0.4\linewidth}
        \centering
    \includegraphics[width=0.8\linewidth]{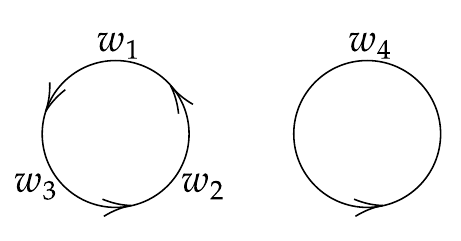}
    \caption{The weighted cycle type of \cref{fig-4110125}.}
    \end{subfigure}
    \caption{A weighted permutation with the corresponding weighted cycle type.}\label{fig-4110128}
\end{figure}

In general, the involution does not preserve the weighted cycle type of a weighted permutation $X\in\mathfrak{W}_n$ (an example is shown in \cref{fig-3260334}), which behaves differently from the ordinary cycle types of permutations. 
Nevertheless, it does preserve the weighted cycle type when the total weight $|X|\leq 2$ is small.

\begin{lemma}\label{lemma-4121756}
Given a weighted permutation $X\in\mathfrak{W}_n$, if the total weight $|X| \leq 2$, then the weighted cycle type of $X$ is the same as that of $X^\dag$.
\end{lemma}
\begin{proof}
Suppose $X=\pi w$. Then, we have $|w|_1\leq 2$.
The involution $\dag$ only reverses the edge directions in each directed cycle graph. It suffices to prove that a directed cycle graph is invariant under reversing when the sum of its weights is no more than $2$. 

If the sum of weights is $0$, then this is a cycle with all edges having weight $0$, and reversing its edge directions gives the same cycle.

If the sum of weights is $1$, then this is a cycle with only one edge having weight $1$ and others having weight $0$, and reversing its edge directions gives the same cycle.

If the sum of weights is $2$ with only one edge having weight $2$ and others having weight $0$, then this is the same as in the previous situation. If there are two edges having weight $1$, suppose that, starting from one of these two edges, we need to traverse $n_1$ weight-$0$ edges to get to the second weight-$1$ edge, and then traverse $n_2$ weight-$0$ edges to get back to the first weight-$1$ edge. Then, if the edge directions are reversed, we can start from the second weight-$1$ edge and traverse $n_1$ weight-$0$ edges to get to the first weight-$1$ edge, and then traverse $n_2$ weight-$0$ edges to get back to the second weight-$1$ edge. Therefore, the reversed cycle graph is the same as the original cycle graph.
\end{proof}

For convenience, we will use the following notation.
\begin{definition}\label{def-4130413}
    For every weighted cycle type $\tau \vdash n$, define $c_{\tau} \colon \mathbb{C}\mathfrak{W}_n \to \mathbb{C}$ such that for any $\mathcal{X}\in\mathbb{C}\mathfrak{W}_n$, 
    \[c_\tau(\mathcal{X})\coloneqq \sum_{X\textup{ has type }\tau} c_X(\mathcal{X}),\]
    where $c_X(\cdot)$ is defined in \cref{def-4130359}.
\end{definition}

\subsection{Simultaneous estimators for \texorpdfstring{$\tr(\mathcal{O}\rho^k)$}{tr(O rho\^k)}}\label{sec-2111442}
Given an observable $\mathcal{O}\in\mathcal{L}(\mathcal{H})$, let $\mu$ be the matrix representation of $\mathfrak{W}_n$ based on $\mathcal{O}$ as defined in \cref{def-4120344}. We define the following symmetrization map.
\begin{definition}\label{def-4131141}
Let $\Phi$ be the symmetrization map acting on $\mathbb{C}\mathfrak{W}_n$ as: 
\[\Phi(X)= \frac{1}{n!}\sum_{\pi\in\mathfrak{S}_n} \pi X \pi^{-1},\]
where $X\in \mathbb{C}\mathfrak{W}_n$.
\end{definition}
Then, our estimators are defined as follows.
\begin{definition}\label{def-2110411}
For any \(1\leq k\leq  n\), we define
\begin{equation}\label{eq-2132257}
\mathcal{O}_k\coloneqq \mu(\Phi(s_k e_1)),
\end{equation}
where \(s_k\in\mathfrak{S}_n\) is the cyclic shift permutation on the first \(k\) elements, i.e., $s_k(i)=\rbra{i \bmod k} +1$ for $i\leq k$ and $s_k(i)=i$ for $k< i\leq n$, and $e_1=[1,0,\ldots,0]\in\mathbb{Z}_{\geq 0}^n$.
\end{definition}

We have the following results.
\begin{proposition}\label{prop-4121906}
\(\mathcal{O}_k\) is Hermitian and \(\tr(\mathcal{O}_k\rho^{\otimes n})=\tr(\mathcal{O}\rho^k)\).
\end{proposition}
\begin{proof}
Note that
\begin{align}
\mathcal{O}_k^\dag & = \mu\big(\Phi(s_k e_1)^\dag\big)= \mu\big(\Phi\big((s_k e_1)^\dag\big)\big).
\end{align}
Since $|e_1|_1=1\leq 2$, by \cref{lemma-4121756}, $(s_k e_1)^\dag$ and $s_k e_1$ have the same weighted cycle type (thus the same orbit). Note that $\Phi(s_k e_1)$ produces the average of elements in the orbit of $s_k e_1$ and thus equals $\Phi((s_k e_1)^\dag)$. Therefore, $\mathcal{O}_k^\dag=\mathcal{O}_k$.

On the other hand, we have
\begin{align}
\tr(\mathcal{O}_k\rho^{\otimes n})&=\frac{1}{n!}\sum_{\pi\in\mathfrak{S}_n}\tr\Big(U_\pi U_{s_k} (\mathcal{O}\otimes I^{\otimes n-1}) U_{\pi}^\dag \rho^{\otimes n}\Big) \nonumber \\
&=\frac{1}{n!} \sum_{\pi\in\mathfrak{S}_n} \tr\Big(U_{s_k} (\mathcal{O}\otimes I^{\otimes n-1}) \rho^{\otimes n} \Big) \nonumber\\
&=\tr(\mathcal{O}\rho^k).\nonumber
\end{align}
\end{proof}

\begin{proposition}\label{lemma-2110242}
For any \(1\leq i\leq j \leq n\),
\[\mathcal{O}_i \mathcal{O}_j=\mathcal{O}_j\mathcal{O}_i.\]
\end{proposition}
\begin{proof}
By \cref{prop-4121906}, we know that $\mathcal{O}_i^\dag=\mathcal{O}_i$, which means $\mathcal{O}_j\mathcal{O}_i=(\mathcal{O}_i\mathcal{O}_j)^\dag$. Let $P_i$ denote $\Phi(s_i e_1)$. Thus, $\mathcal{O}_i\mathcal{O}_j=\mu(P_i)\mu(P_j)=\mu(P_i P_j)$. Then, it suffices to show that  
\[P_i P_j= (P_i P_j)^\dag.\]

Suppose 
\[P_iP_j=\sum b_X X,\]
where $X\in\mathfrak{W}_n$ and $b_X=0$ for all but only finitely many $X$. From the definition of $P_i$, it is obvious that $b_X$ are real numbers. Then, $(P_iP_j)^\dag=\sum b_X X^\dag$. Furthermore, by the definition of $P_i$ and \cref{fact-4130216}, we know that for all $X$ such that $b_X\neq 0$, we have $|X|=2$, and thus $X^\dag$ has the same weighted cycle type as $X$ by \cref{lemma-4121756}. Therefore, for any weighted cycle type $\tau\vdash n$, we have 
\begin{equation}\label{eq-4130409}
c_\tau(P_iP_j) = c_\tau((P_iP_j)^\dag),
\end{equation}
where $c_\tau(\cdot)$ is defined in \cref{def-4130413}.

On the other hand, since all $P_i$ are permutation-invariant by definition (i.e., $\pi P_i \pi^{-1}=P_i$ for any $\pi\in\mathfrak{S}_n$), $P_i P_j$ is also permutation-invariant. Note that the conjugation action of $\mathfrak{S}_n$ acts transitively on each orbit in $\mathfrak{W}_n$. That is, for any $X,Y\in\mathfrak{W}_n$, if $X,Y$ are in the same orbit (of the same weighted cycle type), there exists a $\pi\in\mathfrak{S}_n$ such that
\[\pi X\pi^{-1}= Y.\]
Further note that the conjugation map $\pi (\cdot) \pi^{-1}$ is injective. We can conclude that $X$ and $Y$ must have the same coefficient in $P_iP_j$, i.e., $c_{X}(P_iP_j)=c_{Y}(P_iP_j)$. Therefore, if $X$ is of type $\tau$, then $c_X(P_iP_j)=c_\tau(P_i P_j)/ |\tau|$, where $|\tau|$ is the number of weighted permutations in the orbit $\tau$.
Since $(P_i P_j)^\dag$ is also permutation-invariant, the same argument works for $(P_iP_j)^\dag$, thus we have $c_X((P_iP_j)^\dag)= c_{\tau}((P_iP_j)^\dag)/|\tau|$. 
Then, for any $X\in\mathfrak{W}_n$, letting $\tau$ be its weighted cycle type, we have
\[c_X((P_iP_j)^\dag)= c_{\tau}((P_iP_j)^\dag)/|\tau|=c_\tau(P_i P_j)/ |\tau|=c_X(P_iP_j),\]
where the second equality is by \cref{eq-4130409}.
Therefore, $P_iP_j=(P_iP_j)^\dag$.
\end{proof}

\subsection{Bounding the variance}\label{sec-4271217}
Now, we fix \(\rho\) and introduce the following notations. Suppose \(\mathcal{Q}\) is an observable acting on \(\mathcal{H}^{\otimes n}\), let \(x\) be the outcome of measuring \(\rho^{\otimes n}\) with the observable \(\mathcal{Q}\). Then we write \(\E[\mathcal{Q}]\coloneqq \E[x]\) and \(\Var[\mathcal{Q}]\coloneqq \Var[x]\). It is easy to see that
\[\E[\mathcal{Q}]=\tr(\mathcal{Q}\rho^{\otimes n}),\]
\[\Var[\mathcal{Q}]=\E[\mathcal{Q}^2]-\E[\mathcal{Q}]^2=\tr(\mathcal{Q}^2\rho^{\otimes n})-\tr(\mathcal{Q}\rho^{\otimes n})^2.\]

Let $\mathcal{O}_k$ be the observables defined in \cref{def-2110411}. We will bound \(\Var[\mathcal{O}_k]\) in this section. 
For a sequence \(J=(j_1,j_2,\ldots,j_l)\subseteq [n]\), let \(s_J\in\mathfrak{S}_n\) be the cyclic shift permutation acting on \(J\), i.e., \(s_J\) maps \(j_i\) to \(j_{i+1}\), maps \(j_{l}\) to \(j_1\) and leaves other elements unchanged. Let
\begin{align*}
    T_k\coloneqq \frac{1}{2\floor{n/k}}\sum_{i=0}^{\floor{n/k}-1} &\Big(\mu(s_{(ik+1,\ldots,ik+k)}e_{ik+1})
    +\mu(s_{(ik+1,\ldots,ik+k)}e_{ik+1})^\dag\Big),
\end{align*}
where $e_{i}=[0,\ldots,1,\ldots,0]\in\mathbb{Z}_{\geq 0}^n$ is the vector with a $1$ in the $i$-th place and $0$'s elsewhere.
We have the following result.
\begin{lemma}\label{lemma-4131316}
\[\Var[\mathcal{O}_k]\leq \Var[T_k].\]
\end{lemma}
\begin{proof}
The proof follows the idea in \cite[Lemma 3.14]{BOW19}.

By symmetry, for any \(i\), we have
\[\Phi(s_{(ik+1,\ldots,ik+k)}e_{ik+1})=\Phi(s_{k}e_1),\]
where $\Phi$ is the symmetrization map defined in \cref{def-4131141}.
Therefore, 
\[\Phi(T_k)=\frac{1}{2} (\mathcal{O}_k+\mathcal{O}_k^\dag) = \mathcal{O}_k,\]
where with a slight abuse of notation, $\Phi$ is reinterpreted as the symmetrization map acting on $\mathcal{L}(\mathcal{H})$, i.e., $\Phi(A)=\frac{1}{n!}\sum_{\pi\in\mathfrak{S}_n}U_\pi A U_\pi^\dag$ for any $A\in\mathcal{L}(\mathcal{H}^{\otimes n})$.
Note that \(\Phi\) is both positive and unital, by the Kadison-Schwarz inequality~\cite{kadison1952generalized}, we have
\[\Phi(T_k)^2\sqsubseteq \Phi(T_k^2),\]
where \(\sqsubseteq\) is the Loewner order. This means
\begin{align}
\E[\mathcal{O}_k^2]&= \tr(\mathcal{O}_k^2\rho^{\otimes n})=\tr(\Phi(T_k)^2\rho^{\otimes n})\nonumber \\
&\leq  \tr(\Phi(T_k^2)\rho^{\otimes n})\nonumber\\
&=\tr(T_k^2\rho^{\otimes n}) \label{eq-270402}\\
&=\E[T_k^2],\nonumber
\end{align}
where \cref{eq-270402} is because \(\pi\in\mathfrak{S}_n\) commutes with \(\rho^{\otimes n}\). Since \(\E[\mathcal{O}_k]=\tr(\mathcal{O}\rho^{k})=\E[T_k]\), we have \(\Var[\mathcal{O}_k]\leq \Var[T_k]\).
\end{proof}

Then, it remains to upper bound \(\Var[T_k]\). We have the following result.
\begin{lemma}\label{lemma-4131315}
\[\Var[T_k]\leq \frac{2k\|\mathcal{O}\|^2}{n}.\]
\end{lemma}
\begin{proof}
Let $S_i$ denote $\mu(s_{(ik+1,\ldots,ik+k)}e_{ik+1})+\mu(s_{(ik+1,\ldots,ik+k)}e_{ik+1})^\dag$. 
Note that for \(i\neq j\), $S_i$ and $S_j$ act non-trivially on different subsystems, and thus
\begin{align}
&\E\sbra*{S_i S_j} = \E\sbra*{S_i}\E\sbra*{S_j}. \nonumber
\end{align} 
Therefore
\begin{align}
\Var[T_k]&=\Var\sbra*{\frac{1}{2\lfloor n/k\rfloor}\sum_{i=0}^{\lfloor n/k \rfloor-1} S_i }\nonumber \\
&=\frac{1}{4\lfloor n/k\rfloor^2}\sum_{i=0}^{\lfloor n/k\rfloor -1}\Var\sbra*{S_i}\nonumber \\
&\leq \frac{1}{4 \lfloor n/k\rfloor^2} \cdot \lfloor n/k \rfloor\cdot 4\|\mathcal{O}\|^2 \label{eq-270416}\\
&= \frac{\|\mathcal{O}\|^2}{\lfloor n/k\rfloor} \nonumber \\
&\leq \frac{2k\|\mathcal{O}\|^2}{n},\nonumber
\end{align}
where in \cref{eq-270416} we use the fact that the operator norm of \(S_i\) is upper bounded by \(2\|\mathcal{O}\|\).
\end{proof}

Then, we have the following result.
\begin{proposition}\label{prop-4131317}
\[\Var[\mathcal{O}_k]\leq \frac{2k\|\mathcal{O}\|^2}{n}.\]
\end{proposition}
\begin{proof}
This is by combining \cref{lemma-4131316} with \cref{lemma-4131315}.
\end{proof}

\subsection{Proof of \texorpdfstring{\cref{thm-2151640}}{Theorem III.1}}\label{sec-4131102}
\begin{proof}[Proof of \cref{thm-2151640}]
We define \(p_1,\ldots,p_n\) to be the outcomes of measuring \(\rho^{\otimes n}\) with the observables \(\mathcal{O}_1,\ldots,\mathcal{O}_n\). Note that this is well-defined as \(\mathcal{O}_i\) are valid observables and pairwise commuting due to \cref{prop-4121906} and \cref{lemma-2110242}. Thus, we can perform the measurements simultaneously. Then, by \cref{prop-4121906} and \cref{prop-4131317}, we have \(\E[p_k]=\E[\mathcal{O}_k]=\tr(\mathcal{O}\rho^k)\) and \(\Var[p_k]=\Var[\mathcal{O}_k]\leq 2k\|\mathcal{O}\|^2/n\).
\end{proof}

\section{Sample Complexity Upper Bounds} \label{sec:upper}

\subsection{Nonlinear functionals}

\begin{theorem} \label{thm:upper}
    For any known observable $\mathcal{O}$, we can simultaneously estimate $\tr\rbra{\mathcal{O}\rho}, \tr\rbra{\mathcal{O}\rho^2},\dots, \tr\rbra{\mathcal{O}\rho^k}$ to within additive error $\varepsilon$ using $O\rbra{k\log\rbra{k} \Abs{\mathcal{O}}^2/\varepsilon^2}$ samples of $\rho$, where $\Abs{\cdot}$ is the operator norm. 
\end{theorem}

\begin{proof}
    From \cref{thm-2151640}, given \(n\) samples of an unknown state \(\rho\), we can output a list $p_1, p_2, \ldots, p_n $ such that for any \(1\leq k\leq n\),
    \[\E[p_k]=\tr(\mathcal{O} \rho^k) \quad\quad\textup{and}\quad\quad \Var[p_k]\leq \frac{2k\|\mathcal{O}\|^2}{n}.\] Using Chebyshev's inequality, we have
    \[
    \Pr\sbra{\abs{p_k - \E\sbra{p_k}} \geq \varepsilon} \leq \frac{2k\|\mathcal{O}\|^2}{n \varepsilon^2}.
    \]
    We take $n = \ceil{6k\|\mathcal{O}\|^2/\varepsilon^2}$ to ensure that the probability \[\Pr\sbra{\abs{p_j - \E\sbra{p_j}} \geq \varepsilon} \leq \frac{2 j \|\mathcal{O}\|^2}{n \varepsilon^2} \leq \frac{1}{3}\] for all $1 \leq j \leq k$. Then we can use the median trick to further boost the confidence of estimators $p_j$ to at least $1 - {1}/\rbra{3k}$ for all $1 \leq j \leq k$ with an $O\rbra{\log\rbra{k}}$ overhead.
    
    Specifically, we repeat the above estimation process $m$ times and obtain the random variables $p_j^{(1)}, p_j^{(2)}, ..., p_j^{(m)}$ for all $1 \leq j \leq k$. Then we pick the median  $\tilde{p}_j = \operatorname{median}\rbra{p_j^{(1)}, p_j^{(2)}, ..., p_j^{(m)}}$ as the estimator of $\tr\rbra{\mathcal{O}\rho^j}$. Note that each $p_j^{(\ell)}$ for $1 \leq \ell \leq m$ is an independent random variable such that  
    \[\Pr\sbra*{\abs*{p_j^{\rbra{\ell}} - \tr\rbra{\mathcal{O}\rho^j}} \geq \varepsilon} \leq \frac{1}{3}.\]
    Define indicator variables $Y_{j, \ell} = \mathbbm{1}_{\cbra{\abs{p_j^{\rbra{\ell}} - \tr\rbra{\mathcal{O}\rho^j}} \geq \varepsilon}}$, so $Y_{j, \ell} \sim \operatorname{Bernoulli}\rbra{p}$ with $p \leq 1/3$. Let $Y_j = \sum_{\ell=1}^{m} Y_{j, \ell}$ be the number of ``bad'' estimates among the $m$ trials. To bound the probability that more than half of the $Y_{j, \ell}$'s are bad,
    \[ \Pr\sbra{\abs{\tilde{p}_j - \tr\rbra{\mathcal{O}\rho^j}}> \varepsilon}  \leq \Pr\sbra{Y_j > m/2}.\]
    Now apply Hoeffding's inequality~\cite[Theorem 2]{Hoe63} with $\E\sbra{Y_j} = mp \leq m/3$, we have
    \[ \Pr\sbra{Y_j > m/2} = \Pr\sbra*{Y_j - \E\sbra{Y_j} > m/6} \leq \exp\rbra{-m /18}. \]
    By choosing $m = \ceil{18 \ln\rbra{3k}}$, we have the success probability
    \[
    \Pr\sbra*{ \abs*{\tilde p_j - \tr\rbra{\mathcal{O}\rho^j}} \leq \varepsilon } \geq 1 - \frac{1}{3k}.
    \]
    
    Then by the union bound, we can simultaneously estimate $\tr\rbra{\mathcal{O}\rho}, \tr\rbra{\mathcal{O}\rho^2},\dots, \tr\rbra{\mathcal{O}\rho^k}$ to within additive error $\varepsilon$ with probability at least $2/3$. The entire estimation process uses $n\cdot m = \ceil{6k\|\mathcal{O}\|^2/\varepsilon^2} \cdot \ceil{18 \ln\rbra{3k}} =  O\rbra{k\log\rbra{k} \Abs{\mathcal{O}}^2/\varepsilon^2}$ samples of $\rho$.
\end{proof}

\subsection{Extension to general functionals by polynomial approximation} \label{sec:extension}
Our method enables us to simultaneously estimate several functionals of the form $\tr\rbra{\mathcal{O}g\rbra{\rho}}$, where $g$ is a real function.

\begin{corollary} \label{corollary:poly-func}
    Let $\mathcal{O}$ be an observable and let $g_1, g_2, \dots, g_m \colon \sbra{0, 1} \to \mathbb{R}$ be functions that can be approximated respectively by degree-$k$ polynomials $f_1, f_2, \dots, f_m$ to precision $\varepsilon/(2\Abs{\mathcal{O}}d)$.
    For any $d$-dimensional state $\rho$, we can simultaneously estimate $\tr\rbra{\mathcal{O}g_1\rbra{\rho}}$\textup{,} $\tr\rbra{\mathcal{O}g_2\rbra{\rho}}$\textup{,} \dots\textup{,} $\tr\rbra{\mathcal{O}g_m\rbra{\rho}}$ to within additive error $\varepsilon$ using 
    \[
    O\rbra*{\frac{k\Abs{\mathcal{O}}^2\log\rbra{\min\cbra{k, m}}}{\varepsilon^2}\max_{1 \leq i \leq m}\Abs{f_i}_1^2}
    \]
    samples of $\rho$, where $\Abs{f_i}_1$ is the $\ell_1$-norm of the polynomial coefficients of $f_i$. 
\end{corollary}

\cref{corollary:poly-func} improves and generalizes the results in \cite{QKW24} (see also \cite{YLLW24}), where they presented an approach for the case where $\mathcal{O} = I$ and $m = 1$ with sample complexity $O\rbra{k^2\Abs{f_1}_1^2/\varepsilon^2}$. 
In particular, \cref{corollary:poly-func} improves it to $O\rbra{k\Abs{f_1}_1^2/\varepsilon^2}$ by a factor of $k$, and if there are $m$ different functionals to estimate, it incurs only an overhead of $\log\rbra{\min\cbra{k, m}}$. 
Moreover, \cref{corollary:poly-func} is optimal even when $m = 1$, $g_1\rbra{x} = f_1\rbra{x} \propto x^k$, and $\mathcal{O} \propto I$, where $\Omega\rbra{k\Abs{\mathcal{O}}^2 \Abs{f_1}_1^2 / \varepsilon^2}$ samples of $\rho$ are required as implied by \cref{thm:optimal}. 

\begin{proof}[Proof of \cref{corollary:poly-func}]
    Let us first focus on a single polynomial 
    $f\rbra{x}=\sum_{j=1}^k \alpha_j x^j$.
    For a quantum state $\rho$,
    using \cref{thm-2151640},
    we can obtain $p_1,\ldots,p_k$ such that for $j=1$ to $k$,
    \begin{equation*}
        \E\sbra*{p_j}=\tr\rbra*{\mathcal{O}\rho^j}, \text{ and } \Var\sbra*{p_j}\leq\frac{2\Abs*{\mathcal{O}}^2 j}{n}.
    \end{equation*}
    Suppose that $X_1,\ldots, X_k$ are random variables. Observe that the standard deviation of the random variable $X=X_1+\ldots+ X_k$ can be upper bounded by the sum of the standard deviations of all $X_1,\ldots,X_k$ as follows:
    \begin{align*}
    \sigma[X_1+\ldots +X_k]
    = \, &\sqrt{\Var[X_1+\ldots+X_k]}\\
    = \, &\sqrt{\sum_j \Var[X_j]+ \sum_{j \neq l}\Cov[X_j,X_l]}\\
    \leq \, & \sqrt{\sum_j \Var[X_j]+ \sum_{j \neq l}\sqrt{\Var[X_j]\Var[X_l]}}\\
    = \, & \sqrt{\rbra*{\sum_j \sigma[X_j]}^2}\\
    = \, & \sum_j \sigma [X_j],
    \end{align*}
    where $\Cov[\cdot, \cdot]$ denotes the covariance between two random variables, and we use the Cauchy–Schwarz inequality.
    Let $\hat{p}=\sum_{j=1}^k \alpha_j p_j$.
    Then, we can calculate $\E\sbra*{\hat{p}}=\tr\rbra*{\mathcal{O}f(\rho)}$ and
    \begin{equation*}
        \sigma\sbra*{\hat{p}}\leq \!\sum_{j=1}^k \abs{\alpha_j}\cdot \sigma\sbra*{p_j}\leq \sum_{j=1}^k \!\sqrt{\frac{2 \Abs*{\mathcal{O}}^2 j}{n}}\abs{\alpha_j}\leq\! \sqrt{\frac{2\|O\|^2k}{n}}\|f\|_1,
    \end{equation*}
    by the above observation.
    Consequently, due to Chebyshev's inequality, setting the number of samples
    \begin{equation*}
        n=O\rbra*{\frac{k \Abs*{f}_1^2\Abs*{\mathcal{O}}^2}{\varepsilon^2}}
    \end{equation*} 
    allows us to estimate $\tr\rbra*{\mathcal{O}f(\rho)}$
    to within additive error $\varepsilon/2$ and with probability $\geq 2/3$.
    Like in the proof of \cref{thm:upper}, by repeating this experiment $O\rbra*{\log(1/\delta)}$ times, we can amplify the success probability to $1-\delta$.

    Now consider multiple polynomials $f_1,f_2,\ldots,f_m$. 
    \begin{itemize}
        \item 
        If $m\leq k$, then for each $i=1$ to $m$, we can estimate $\tr\rbra*{\mathcal{O}f_i(\rho)}$ to within additive error $\varepsilon/2$ and with probability $\geq 1-1/(3m)$.
        By the union bound, with overall success probability $\geq 2/3$, we can simultaneously estimate all $\tr\rbra*{\mathcal{O}f_1(\rho)}, \ldots, \tr\rbra*{\mathcal{O}f_m(\rho)}$ to within error $\varepsilon/2$. 
        \item 
        If $m> k$, then for each $j=1$ to $k$, we can estimate $\tr\rbra*{\mathcal{O}\rho^j}$  to within additive error $\varepsilon/(2\max_{i=1}^m\Abs*{f_i}_1)$ and with probability $\geq 1-1/(3k)$. 
        By the union bound, with success probability $\geq 2/3$, for any $i=1$ to $m$, 
        the linear combination of estimates of $\tr\rbra*{\mathcal{O}\rho^j}$ gives an estimation of $\tr\rbra*{\mathcal{O}f_i(\rho)}$ to within error $\varepsilon/2$.
    \end{itemize}
    Combining the results above, we can simultaneously estimate all $\tr\rbra*{\mathcal{O}f_1(\rho)}, \ldots, \tr\rbra*{\mathcal{O}f_m(\rho)}$ to within additive error $\varepsilon/2$, with overall success probability $\geq 2/3$, using \[O\rbra{k\Abs{\mathcal{O}}^2\max_{i=1}^m\Abs{f_i}_1^2\log\rbra{\min\cbra{k, m}}/\varepsilon^2}\] samples of $\rho$.
    Since in \Cref{corollary:poly-func},
    we assume that each $f_i(x)$ approximates $g_i(x)$ to precision $\varepsilon/(2\Abs*{\mathcal{O}}d)$ for $x\in [0,1]$, with probability $\geq2/3$, our estimates of $\tr\rbra*{\mathcal{O}f_i(\rho)}$ approximate $\tr\rbra*{\mathcal{O}g_i(\rho)}$ to within error $\varepsilon$.
    The conclusion immediately follows.
\end{proof}

\section{Sample Complexity Lower Bounds} \label{sec:lower}

In this section, we first establish a matching lower bound on the sample complexity of estimating $\tr\rbra{\rho^k}$.
Then, we extend the proof idea to show a lower bound on the sample complexity of estimating $\tr\rbra{\mathcal{O}\rho^k}$,
where $\mathcal{O}$ is a given observable. 
Further, we derive lower bounds on the sample complexity of estimating 
$\tr\rbra{f\rbra{\rho}}$ and $\tr\rbra{\mathcal{O}f\rbra{\rho}}$
for a general functional $f$ approximated by polynomials.

\subsection{Estimation of quantum state powers \texorpdfstring{$\tr\rbra{\rho^k}$}{tr(rho\^k)}}

Let us start with the lower bound on the sample complexity of estimating $\tr\rbra{\rho^k}$;
that is, the special case $\mathcal{O}=I$ in \cref{thm:optimal}.
Before proceeding, we introduce the following theorem and fact about quantum state discrimination. 

\begin{theorem} [Quantum state discrimination, cf.\ {\cite[Section 9.1.4]{Wil13}} and {\cite[Lemma 3.2]{Hay17}}]
    Let $\rho_0$ and $\rho_1$ be two quantum states, and let $\rho$ be a quantum state such that $\rho = \rho_0$ or $\rho = \rho_1$ with equal probability.
    Then, any POVM $\Lambda = \cbra{\Lambda_0, \Lambda_1}$ determines whether $\rho = \rho_0$ or $\rho = \rho_1$ with success probability at most 
    \begin{equation*}
        \frac{1}{2} + \frac{1}{4} \Abs{\rho_0 - \rho_1}_1,
    \end{equation*}
    where 
    \begin{equation*}
        \frac{1}{2} \Abs{\rho_0 - \rho_1}_1 = \frac{1}{2} \tr\rbra{\abs{\rho_0-\rho_1}}
    \end{equation*}
    is the trace distance.
\end{theorem}

\begin{fact} \label{fact:qsd}
    Any quantum algorithm that distinguishes between two quantum states $\rho_0$ and $\rho_1$ requires sample complexity $\Omega\rbra{1/\gamma}$, where $\gamma = 1 - \mathrm{F}\rbra{\rho_0, \rho_1}$ is the infidelity and 
    \begin{equation*}
        \mathrm{F}\rbra{\rho_0, \rho_1} = \tr\rbra*{\sqrt{\sqrt{\rho_0}\rho_1\sqrt{\rho_0}}}
    \end{equation*}
    is the fidelity. 
\end{fact}

Now we can prove \cref{thm:optimal} for the special case of $\mathcal{O} = I$.  

\begin{theorem} \label{thm:O=I}
    For sufficiently large integer $k$ and sufficiently small $\varepsilon > 0$, estimating $\tr\rbra{\rho^k}$ to within additive error $\varepsilon$ requires sample complexity $\Omega\rbra{k/\varepsilon^2}$.
\end{theorem}

The lower bound $\Omega\rbra{k/\varepsilon^2}$ in \cref{thm:O=I} considers the dependence on $k$, whereas the previous lower bound $\Omega\rbra{1/\varepsilon^2}$ in \cite{CW25,CLW26} assumes $k = \Theta\rbra{1}$. 

\begin{proof}[Proof of \cref{thm:O=I}]
Suppose that estimating $\tr\rbra{\rho^k}$ to additive error $\varepsilon$ can be done with sample complexity $S\rbra{k, \varepsilon}$. 
Let $\varepsilon \in \rbra{0, 1/2}$. 
Consider the problem of distinguishing the two quantum states $\rho_\pm$ defined by
\begin{equation}
    \label{eq:rho-pm}
    \rho_\pm = \rbra*{1-\frac{1}{k} \pm \frac{\varepsilon}{k}} \ketbra{0}{0} + \rbra*{\frac{1}{k}\mp\frac{\varepsilon}{k}} \ketbra{1}{1}. 
\end{equation}
Note that
\begin{equation*}
\begin{split}
    &\tr\rbra{\rho_+^k} -\tr\rbra{\rho_-^k}= \rbra*{1-\frac{1}{k} + \frac{\varepsilon}{k}}^k + \rbra*{\frac{1}{k}  - \frac{\varepsilon}{k}}^k 
-\rbra*{1-\frac{1}{k} - \frac{\varepsilon}{k}}^k - \rbra*{\frac{1}{k} + \frac{\varepsilon}{k}}^k.
\end{split}
\end{equation*}
Then, we have
\begin{equation} \label{eq:tr-diff}
    \tr\rbra{\rho_+^k} - \tr\rbra{\rho_-^k} \geq \Omega\rbra{\varepsilon}.
\end{equation}
To see this, let us consider the following function
\begin{align*}
&f(x)=\rbra*{1-\frac{1}{k} + \frac{x}{k}}^k + \rbra*{\frac{1}{k}-\frac{x}{k}}^k  
 -\rbra*{1-\frac{1}{k} - \frac{x}{k}}^k - \rbra*{\frac{1}{k}+\frac{x}{k}}^k
\end{align*}
and then let $x\in (0,1/2)$ and consider the derivative
\begin{align*}
\frac{\partial}{\partial x} f(x) &= \rbra*{1-\frac{1}{k}+\frac{x}{k}}^{k-1} -  \rbra*{\frac{1}{k}-\frac{x}{k}}^{k-1}  
+  \rbra*{1-\frac{1}{k}-\frac{x}{k}}^{k-1} -  \rbra*{\frac{1}{k}+\frac{x}{k}}^{k-1} \\
    & \geq \rbra*{1-\frac{1}{k}-\frac{x}{k}}^{k-1} -  \rbra*{\frac{1}{k}+\frac{x}{k}}^{k-1} \\
    & \geq \rbra*{1-\frac{2}{k}}^{k-1} -  \rbra*{\frac{2}{k}}^{k-1} \\
    & \geq \Omega\rbra{1},
\end{align*}
where the last inequality is because
\begin{equation*}
    \lim_{k \to \infty} \rbra*{ \rbra*{1-\frac{2}{k}}^{k-1} - \rbra*{\frac{2}{k}}^{k-1} } = \frac{1}{e^2}. 
\end{equation*}
Then, by the mean value theorem, there exists $0< \delta<\varepsilon$ such that
\[\tr(\rho_+^k)-\tr(\rho_-^k)=f(\varepsilon)-f(0)= \varepsilon \cdot  \left(\left . \frac{\partial}{\partial x} f(x) \right|_{x=\delta}\right)\geq \Omega(\varepsilon).\]

Then, we can distinguish $\rho_+$ and $\rho_-$ by separately estimating $\tr\rbra{\rho_+^k}$ and $\tr\rbra{\rho_-^k}$ to within additive error $\Theta\rbra{\varepsilon}$, each with sample complexity $S\rbra{k, \Theta\rbra{\varepsilon}}$. 

On the other hand, by \cref{fact:qsd}, the quantum sample complexity of distinguishing $\rho_+$ and $\rho_-$ is 
\begin{equation*}
    \Omega\rbra*{\frac{1}{1-\mathrm{F}\rbra{\rho_+, \rho_-}}} \geq \Omega\rbra*{\frac{k}{\varepsilon^2}}
\end{equation*}
for $\varepsilon \in (0,1/2)$. 
To this end, we first note that
\[1-\mathrm{F}\rbra{\rho_+, \rho_-} \!=\! 1 - \sqrt{\rbra*{1-\frac{1}{k}}^2 - \frac{\varepsilon^2}{k^2}} - \sqrt{\rbra*{\frac{1}{k}}^2 - \frac{\varepsilon^2}{k^2}}.\]
Then, let $x\in (0,1/4)$ and consider the following derivative
\begin{align*}
   \frac{\partial}{\partial x} \rbra*{1 - \sqrt{\rbra*{1-\frac{1}{k}}^2 - \frac{x}{k^2}} - \sqrt{\rbra*{\frac{1}{k}}^2 - \frac{x}{k^2}}}  
 = \, & \frac{1}{2k^2\sqrt{\rbra*{1-\frac{1}{k}}^2 - \frac{x}{k^2}}} + \frac{1}{2k^2\sqrt{\rbra*{\frac{1}{k}}^2 - \frac{x}{k^2}}} \\
    \leq \, & \frac{1}{2k^2\sqrt{\rbra*{1-\frac{1}{k}}^2 - \frac{1}{4k^2}}} + \frac{1}{2k^2\sqrt{\frac{3}{4k^2} }}\\
 \leq \,& O\rbra*{\frac{1}{k^2}}+ O\rbra*{\frac{1}{k}}\\
    \leq \,& O\rbra*{\frac{1}{k}}. 
\end{align*}
Again by the mean value theorem, there exists $0<\delta<\varepsilon^2$ such that
\begin{align*}
1-\mathrm{F}\rbra{\rho_+, \rho_-}
=\, &\varepsilon^2 \cdot \left .\frac{\partial}{\partial x} \rbra*{1 - \sqrt{\rbra*{1-\frac{1}{k}}^2 - \frac{x}{k^2}} - \sqrt{\rbra*{\frac{1}{k}}^2 - \frac{x}{k^2}}}   \! \right|_{x=\delta} \\ 
\leq \, & O\rbra*{\frac{\varepsilon^2}{k}}.
\end{align*}
The above together gives
\begin{equation*}
    S\rbra{k, \Theta\rbra{\varepsilon}} \geq \Omega\rbra*{\frac{1}{1-\mathrm{F}\rbra{\rho_+, \rho_-}}} \geq \Omega\rbra*{\frac{k}{\varepsilon^2}}.
\end{equation*}
By letting $\varepsilon' = \Theta\rbra{\varepsilon}$, we have
\begin{equation*}
    S\rbra{k, \varepsilon'} \geq \Omega\rbra*{\frac{k}{\varepsilon'^2}}.
\end{equation*}
These yield the proof. 
\end{proof}

\subsection{Estimation of \texorpdfstring{$\tr\rbra{\mathcal{O}\rho^k}$}{tr(O rho\^k)} with general observables \texorpdfstring{$\mathcal{O}$}{O}}

Next, we formally restate \cref{thm:optimal} in the following theorem and then show a lower bound on the sample complexity of estimating $\tr\rbra{\mathcal{O}\rho^k}$ for any given observable $\mathcal{O}$. 

\begin{theorem} \label{thm:lower-obs}
    For sufficiently large integer $k$, sufficiently small $\varepsilon > 0$, and any observable $\mathcal{O}$, estimating $\tr\rbra{\mathcal{O}\rho^k}$ to within additive error $\varepsilon$ requires sample complexity $\Omega\rbra{k\Abs{\mathcal{O}}^2/\varepsilon^2}$.
\end{theorem}

\begin{proof}
    Estimating $\tr\rbra{\mathcal{O}\rho^k}$ to within additive error $\varepsilon$ is equivalent to estimating $\tr\rbra{\frac{\mathcal{O}}{\Abs{\mathcal{O}}} \rho^k}$ to within error $\frac{\varepsilon}{\Abs{\mathcal{O}}}$. 
    So, it suffices to assume $\Abs{\mathcal{O}}=1$ and prove that estimating $\tr\rbra{\mathcal{O}\rho^k}$ to error $\varepsilon$ requires $\Omega\rbra{k/\varepsilon^2}$ samples of $\rho$.

    The proof idea is similar to that for \cref{thm:O=I}.
    Suppose that $\rho$ is a $d$-dimensional quantum state.
    Without loss of generality, we set the computational basis $\cbra{\ket{x}}_{x=0}^{d-1}$ such that $\mathcal{O}$ is diagonal in this basis and that $\bra{0}\mathcal{O}\ket{0}=1$.
    Suppose that $\bra{1}\mathcal{O}\ket{1}=a$ for some $a\in [-1,1]$.
    Like the proof of \cref{thm:O=I}, let us consider the problem of distinguishing the two quantum states $\rho_{\pm}$ defined in \cref{eq:rho-pm}.
    We can calculate
    \begin{equation*}
        \tr\rbra*{\mathcal{O}\rho_{\pm}^k}= \rbra*{1-\frac{1}{k}\pm \frac{\varepsilon}{k}}^k + a\rbra*{\frac{1}{k}\mp \frac{\varepsilon}{k}}^k. 
    \end{equation*}
    It can be verified that 
    \begin{align*}
        \abs*{\tr\rbra*{\mathcal{O}\rho_+^k}-\tr\rbra*{\mathcal{O}\rho_-^k}}
        =\, & \rbra*{1-\frac{1}{k} + \frac{\varepsilon}{k}}^k + a\rbra*{\frac{1}{k}-\frac{\varepsilon}{k}}^k
     - \rbra*{1-\frac{1}{k} - \frac{\varepsilon}{k}}^k - a\rbra*{\frac{1}{k}+\frac{\varepsilon}{k}}^k\\
        \geq \, & \rbra*{1-\frac{1}{k} + \frac{\varepsilon}{k}}^k + \rbra*{\frac{1}{k}-\frac{\varepsilon}{k}}^k  - \rbra*{1-\frac{1}{k} - \frac{\varepsilon}{k}}^k - \rbra*{\frac{1}{k}+\frac{\varepsilon}{k}}^k\\
        = \, & \tr\rbra{\rho_+^k} - \tr\rbra{\rho_-^k} \\
        \geq \,& \Omega\rbra{\varepsilon},
    \end{align*}
    where the last inequality is by \cref{eq:tr-diff}.
    Using the same reasoning as in the proof of \cref{thm:O=I},
    the conclusion immediately follows.
\end{proof}

\section{Applications} \label{sec:app}
In this section, we present specific applications of our simultaneous estimator in quantum information. These applications require estimating the entire set $\cbra{\tr\rbra{\mathcal{O}\rho^j}}_{j=1}^k$ rather than a single high-order functional, where the simultaneous estimation yields a quadratic improvement over the one-by-one estimation strategy.

\subsection{Entanglement spectroscopy}
Entanglement spectroscopy~\cite{JST17} is a powerful tool for analyzing the quantum correlations and topological properties of many-body systems. By studying the \emph{entanglement spectrum}, i.e., the eigenvalues of the reduced density matrix $\rho = \tr_B\rbra{\ketbra{\psi}{\psi}}$ obtained from a pure bipartite state $\ket{\psi}$ on systems $A$ and $B$, one could extract a wealth of information about the many-body systems. For example, the eigenvalues of the reduced state $\rho$ diagnose whether the bipartite state $\ket{\psi}$ is entangled or separable and characterize how strong the entanglement is~\cite{HE02}. The entanglement spectrum can be further used to identify topological orders~\cite{Fid10,LH08,PTB+10,
YQ10}, quantum phase transitions~\cite{DLL+12}, many-body localization~\cite{SMA+16,YCH+15,YHG+17}, and irreversibility in quantum systems~\cite{CHM14}. 

Prior works have reduced the task of entanglement spectroscopy to computing high-order functionals of the reduced density matrix, i.e., $\tr\rbra{\rho^k}$ for $k = 2, 3, \ldots,k_{\mathrm{max}}$, which can be done by the Hadamard test~\cite{JST17}, Two-Copy test~\cite{SCC19}, qubit reset~\cite{YS21}, and cyclic shift~\cite{QKW24}. 
Note that $O\rbra{k_{\mathrm{max}}^2 \log\rbra{k_{\mathrm{max}}} /\varepsilon^2}$ samples of $\rho$ are needed in these algorithms to get an estimation of $\{\tr\rbra{\rho^k}\}_{k=2}^{k_{\mathrm{max}}}$ to within additive error $\varepsilon$. 
It can be seen that this can be further improved to $O\rbra{k_{\mathrm{max}} \log\rbra{k_{\mathrm{max}}} /\varepsilon^2}$, using our simultaneous estimator in \cref{thm:main} for the special case where $\mathcal{O} = I$ is the identity operator.

\subsection{Quantum virtual cooling} 
Quantum virtual cooling \cite{CCL+19} provides a way to estimate properties of thermal states at low temperatures, with a prominent application scenario in the doped Fermi-Hubbard model with ultracold atoms \cite{GB17}. 
For a thermal state $\rho\rbra{T} = e^{-\beta H}/\tr\rbra{e^{-\beta H}}$ at temperature $T$ with $H$ the Hamiltonian of the system, where $\beta = 1/\rbra{k_{\textup{B}}T}$ is the inverse temperature and $k_{\textup{B}}$ is the Boltzmann constant, the expectation of an observable $\mathcal{O}$ with respect to the thermal state $\rho\rbra{T/n}$ at the fractional temperature $T/n$ can be expressed as (noted in \cite{CCL+19})
\begin{equation} \label{eq:expect-thermal}
    \tr\rbra{\mathcal{O} \cdot \rho\rbra{T/n}} = \frac{\tr\rbra{\mathcal{O} \cdot \rho\rbra{T}^n}}{\tr\rbra{\rho\rbra{T}^n}}. 
\end{equation}
Thus, one can estimate $\tr\rbra{\mathcal{O} \cdot \rho\rbra{T/n}}$ from the nonlinear functionals $\tr\rbra{\mathcal{O} \cdot \rho\rbra{T}^n}$ and $\tr\rbra{\rho\rbra{T}^n}$ of $\rho\rbra{T}$, using $O\rbra{n}$ samples of $\rho\rbra{T}$. 

To draw a tomographic view of the thermal states of a system at different temperatures, one can consider fractional temperatures such as $T/2$, $T/3$, \dots, $T/n$. 
Specifically, the expectation $\tr\rbra{\mathcal{O}\cdot\rho\rbra{T/k}}$ at temperature $T/k$ can be obtained from $\tr\rbra{\mathcal{O} \cdot \rho\rbra{T}^k}$ and $\tr\rbra{\rho\rbra{T}^k}$ according to \cref{eq:expect-thermal}.
This requires us to estimate the following values:
\[
\begin{matrix}
    \tr\rbra{\mathcal{O} \cdot \rho \rbra{T}^2}, &  \tr\rbra{\mathcal{O} \cdot \rho \rbra{T}^3}, & \dots, & \tr\rbra{\mathcal{O} \cdot \rho \rbra{T}^n}, \\
    \tr\rbra{\rho\rbra{T}^2}, & \tr\rbra{\rho\rbra{T}^3}, & \dots, & \tr\rbra{\rho\rbra{T}^n}.
\end{matrix}
\]
To achieve this, directly employing the approach in \cite{CCL+19} consumes $O\rbra{n^2\log\rbra{n}}$ samples of $\rho$. 
In sharp contrast, the simultaneous estimator in \cref{thm:main} allows us to use only $O\rbra{n\log\rbra{n}}$ samples of $\rho$, yielding a quadratic reduction in resources. 

\section{Discussion} \label{sec:discus}
We present an approach to simultaneously estimating high-order functionals of the form $\tr\rbra{\mathcal{O}\rho^k}$ with almost optimal sample complexity.
We apply this approach in entanglement spectroscopy and quantum virtual cooling, and extend it to general functionals by polynomial approximation. 

An interesting future direction is to extend the simultaneous estimation to more general cases. 

\begin{enumerate}
\item Query-access case. 
In addition to sample access, another input model, known as the ``purified quantum query access'' model \cite{GL20}, gives quantum states via state-preparation circuits, as widely employed in the literature (e.g., \cite{GL20,SH21,GHS21,GMF24,GLM+22,WZC+23,WGL+24,GP22,RASW23,WZYW23,WZL24,Wan24,LW25}). 
However, it is not even clear if we can simultaneously estimate $\tr\rbra{\rho^2}, \dots, \tr\rbra{\rho^k}$ in this query model. 
Existing methods based on the generalized SWAP test \cite{EAO+02} and quantum amplitude estimation \cite{BHMT02} 
estimate $\tr\rbra{\rho^j}$ individually to within additive error $\varepsilon$ with query complexity $O\rbra{j/\varepsilon}$, leading to $O(k^2\log\rbra{k}/\varepsilon)$ for estimating all the $k$ values, which is still worse than the $O\rbra{k\log\rbra{k}/\varepsilon^2}$ given in this paper when $\varepsilon = \Theta\rbra{1}$. \label{ques2}
\item Multivariate case.
The multivariate trace $\tr\rbra{\rho_1\rho_2\dots\rho_m}$ can also be estimated by the cyclic shift \cite{EAO+02}, and was recently solved with constant quantum depth \cite{QKW24}. 
An interesting problem is how to simultaneously estimate $\tr\rbra{\rho\sigma}$, $\tr\rbra{\rbra{\rho\sigma}^2}$, \dots, $\tr\rbra{\rbra{\rho\sigma}^k}$, which is related to the estimation of fidelity \cite{QKW24} and multivariate fidelities \cite{NMLW25}. 
A solution to this problem may inspire new quantum algorithms for estimating quantities involving multiple quantum states, e.g., trace distance \cite{WZ24b} and relative entropy \cite{Hay24}. \label{ques1}
\item Non-integer case. 
Estimating $\tr\rbra{\rho^{\alpha}}$ for non-integer $\alpha$ is the key step in estimating the R\'enyi and Tsallis entropies \cite{LW19,AISW20,WGL+24,WZL24,WZ25b,LW25,CW25,CLW26}. 
An interesting problem is how to estimate $\tr\rbra{\rho^{\alpha_1}}, \tr\rbra{\rho^{\alpha_2}}, \dots, \tr\rbra{\rho^{\alpha_k}}$ for several non-integer $\alpha_1, \alpha_2, \dots, \alpha_k$. 
This could lead to fast simultaneous estimation of entropy. 
This can be more challenging when given query access (with the integer case discussed in Question \ref{ques2}). 
There have been approaches based on quantum singular value transformation \cite{GSLW19} using the polynomial approximation of non-integer power functions, e.g., \cite{WGL+24,WZL24,WZ25b,LW25}. However, it is not clear how to simultaneously deal with multiple approximation polynomials using quantum queries in the framework of quantum singular value transformation.
\item Incoherent case. 
As considered to be near-term friendly, incoherent measurements were investigated for estimating, e.g., $\tr\rbra{\rho\sigma}$ and $\tr\rbra{\rho^2}$ \cite{ACQ22,CCHL22,ALL22,AS25,GHYZ24}.
Recently, estimating $\tr\rbra{\rho^k}$ for large $k$ with incoherent measurements was studied in \cite{PTTW25}.
An interesting question is whether a simultaneous estimation exists in the incoherent case. 
\end{enumerate}

\section*{Acknowledgment}
    The authors would like to thank Mark M.\ Wilde for helpful discussions, Xin Wang and Xiao Shi for communication regarding the related work \cite{SJW+25}, You Zhou for pointing out the related works \cite{ZL24,LLY+24}, and Wang Fang for pointing out a typo in an earlier version of this paper.

    Part of the work of Qisheng Wang was done when the author was with the School of Informatics, University of Edinburgh, Edinburgh, U.K.
    Part of the work of Zhicheng Zhang was done when the author was with the Centre for Quantum Software and Information, University of Technology Sydney, Australia, and part was
    carried out during a visit to DIMACS, Rutgers University, USA.

\addcontentsline{toc}{section}{References}

\bibliographystyle{alphaurl}
\bibliography{main}

\end{document}